\documentclass[10pt,journal,compsoc]{IEEEtran}
\ifCLASSOPTIONcompsoc
  \usepackage[nocompress]{cite}
\else
  \usepackage{cite}
\fi
\ifCLASSINFOpdf
\else
\fi
\usepackage[cmex10]{amsmath}
\usepackage{amssymb}
\usepackage{amsthm}
\usepackage{dsfont} 
\usepackage{paralist} %For compact enumerations
\usepackage[multiple]{footmisc} %for multiple footnotes to be seperated by a comma

\newtheorem{theorem}{Theorem}

\newtheorem{corollary}[theorem]{Corollary}

\newtheorem{remark}{Remark} % Use \normalfont inside the remark in IEEE Journals

\usepackage[ruled,vlined]{algorithm2e}

\usepackage{enumerate} 

\usepackage{graphicx}
\usepackage[outdir=./]{epstopdf}
\usepackage{caption}
\usepackage{subcaption}

\begin{document}
\title{Contagions in Social Networks: Effects of Monophilic Contagion, Friendship Paradox and Reactive Networks}

\author{Buddhika~Nettasinghe,~\IEEEmembership{Student~Member,~IEEE,}
	Vikram~Krishnamurthy,~\IEEEmembership{Fellow,~IEEE,}
	and~Kristina~Lerman% <-this % stops a space
\IEEEcompsocitemizethanks{\IEEEcompsocthanksitem B. Nettasinghe and V. Krishnamurthy are with Cornell~Tech, 2 West Loop~Rd, New~York, NY 10044, USA and, School
	of Electrical and Computer Engineering, Cornell~University, NY, USA.\protect\\
E-mail:~\{dwn26, vikramk\}@cornell.edu
\IEEEcompsocthanksitem K. Lerman is with  Information Sciences Institute, University of Southern California, 4676 Admiralty Way, Marina del Rey, CA 90292, USA.\protect\\
E-mail:~lerman@isi.edu.
}% <-this % stops a space
}

\IEEEtitleabstractindextext{%
\begin{abstract}
We consider SIS contagion processes over networks where, a classical assumption is that individuals' decisions to adopt a contagion are based on their immediate neighbors. However, recent literature shows that some attributes are more correlated between two-hop neighbors, a concept referred to as \textit{monophily}. This motivates us to explore monophilic contagion, the case where a contagion (e.g. a product, disease) is adopted by considering two-hop neighbors instead of immediate neighbors (e.g. you ask your friend about the new iPhone and she recommends you the opinion of one of her friends). We show that the phenomenon called \textit{friendship paradox} makes it easier for the monophilic contagion to spread widely. We also consider the case where the underlying network stochastically evolves in response to the state of the contagion (e.g. depending on the severity of a flu virus, people restrict their interactions with others to avoid getting infected) and show that the dynamics of such a process can be approximated by a differential equation whose trajectory satisfies an algebraic constraint restricting it to a manifold. Our results shed light on how graph theoretic consequences affect contagions and, provide simple deterministic models to approximate the collective dynamics of contagions over stochastic graph processes.
\end{abstract}

}

\maketitle

% is not selected <OR> if conference mode is selected - because compsoc
% conference papers position the abstract like regular (non-compsoc)
\IEEEdisplaynontitleabstractindextext

% \fi
\IEEEpeerreviewmaketitle

\ifCLASSOPTIONcompsoc
\IEEEraisesectionheading{\section{Introduction}\label{sec:introduction}}
\else
\section{Introduction}
\label{sec:introduction}
\fi

\IEEEPARstart{T}{his} spread	 of a contagion (e.g. news, innovations, cultural fads) across a population of agents interconnected by an underlying network is of fundamental interest in many fields including economics, epidemiology, computer science and computational social science.   One of the widely used standard models in the study of such diffusion processes is the Susceptible-Infected-Susceptible (SIS) model \cite{lopez2008diffusion}. In this paper, we focus on a discrete time version of the SIS model on a network which is as follows briefly (reviewed in detail in Sec. \ref{sec:preliminaries}). At each discrete time instant, a randomly sampled individual (called an agent) $m$ from the population observes a $d(m)$ (called the degree of $m$) number of other randomly selected agents (called neighbors of m). Based on the observed sample, the agent m then decides to choose one of the possible states: infected or susceptible. In this context, the aim of this paper is to answer the following  questions:
\begin{compactenum}
	\item what is the effect of some nodes being more likely to change their status than others (e.g. high degree nodes evolving faster than others)? i.e. what is the effect of the distribution with which $m$ is chosen at each time instant? 
	
	\vspace{0.2cm}
	\item what are the effects of two-hop neighbors being considered when making decisions (\textit{monophilic contagion}) and the effect of \textit{friendship paradox} on the contagion process? 
	
	\vspace{0.2cm}
	\item how can one model the collective dynamics of a contagion spreading on a \textit{reactive  network} that evolves in a manner that depends on the state of the contagion?
\end{compactenum}

\subsection{Motivation and Related Work}
\label{subsec:motivation}
The impact of network structure on contagion processes has been studied extensively in literature. Notable works include \cite{lopez2008diffusion,jackson2007relating,lopez2012influence,pastor2001epidemic,jackson2007diffusion,watts2002simple}. Three key assumptions made in most of these work are: 
\begin{compactenum}[i.]
	\item individuals decide whether to adopt the contagion or not based only on their (immediate) neighbors' actions 
	
	\item the underlying social network is fully characterized by its degree distribution

	\item the underlying social network is deterministic and remains same throughout the diffusion process.
\end{compactenum}

However, it has recently been pointed out in \cite{altenburger2018monophily} that certain attributes of individuals might be more similar to friends of friends (referred to as ``the company you are kept in") than to the attributes of their friends (referred to as ``the company you keep"). This phenomenon is referred to as \textit{monophily}. This should be contrasted to \textit{homophily} \cite{mcpherson2001birds} where attributes of individuals are similar to their friends. Hence, motivated by the concept of monophily, our first and second aims (stated previously) model and study the case where, the diffusion process is based on monophilic contagion i.e. agents take their friends of friends into account when presented with the decision to adopt the contagion\footnote{It should be noted that the concept of monophily presented in \cite{altenburger2018monophily} does not give a causal interpretation but only the correlation between two-hop neighbors. What we consider is monophilic contagion (motivated by monophily): the contagion caused by the influence of two hop neighbors.}. 

Secondly, assuming that the network is fully characterized by its degree distribution (assumption 2 stated previously) neglects important characteristics of the joint degree distribution (see for example \cite{wu2018degree} for effects of higher order structural correlations on diffusion) that captures the joint variation of the degrees of neighbors in the social network. Motivated by this, in the second aim of this paper, we also explore how the assortativity (neighbor degree correlation) affects the diffusion process. Further, we also explore how the graph theoretic phenomenon  called \textit{friendship paradox} affects the diffusion process. To the best of our knowledge, the effects of the friendship paradox on the diffusion processes on networks has not been explored in the literature.

Thirdly, we note that most real world networks are of random nature and evolve rapidly during the diffusion process. More generally, the underlying network may evolve in a manner that depends on the state of the diffusion process as well e.g. depending on the state of a spreading disease (fraction of infected individuals for example), people might restrict their interactions with others and thus, changing the structure of the contact network\footnote{\cite{nettasinghe2018influence,clementi2009broadcasting} provide further examples of random graph processes that depend on state of diffusions. One major difference between \cite{nettasinghe2018influence} and the third aim of this paper is that the random graph evolves on the same time scale in the current work while it evolves on a slower (compared to the contagion) time scale in \cite{nettasinghe2018influence}. Further, the context for \cite{nettasinghe2018influence} is independent cascade model which is different from the threshold models studied here.}. Hence, modelling a network as a deterministic graph does not capture this (diffusion state dependent) evolution of real world networks. This serves as the motivation for our third aim where, we model the underlying network as a random graph process (which evolves on the same time scale as the diffusion process) whose transition probabilities at each time instant depend on the state of the  contagion.
\subsection{Main Results and Organization}

{\bf Main results} of this paper are as follows:
\begin{enumerate}
	\item Effect of the random friends (instead of random nodes) evolving at each time step of the discrete time SIS model model is reflected in different elements of the population state evolving at different speeds. However, the minimum spreading rate (defined precisely later) of the contagion required for it to spread to a positive fraction of nodes without dying away (called the critical threshold) is invariant to this change.
	
	\item As a result of \textit{friendship paradox}, the \textit{monophilic contagion} (agents decide whether to adopt a contagion or not by observing random friends of friends) makes it easier for the contagion to prevail i.e. the critical threshold corresponding to monophilic contagion is smaller compared to the non-monophilic case. Further, disassortativity (negative degree-degree correlation coefficient) further amplifies the effect of friendship paradox.
	
	\item If the network is a \textit{reactive network} that randomly evolves depending on the state of the contagion, the collective dynamics of the network and the contagion process can be approximated by an ordinary differential equation (ODE) with an algebraic constraint. From a statistical modeling and machine learning perspective, the importance of this result relies on the fact that it provides a simple deterministic approximation of the collective stochastic dynamics of a complex system (an SIS process on a random graph, both evolving on the same time scale).
\end{enumerate}

\vspace{.25cm}
\noindent
{\bf Organization:} Section \ref{sec:preliminaries} reviews the mean-field approximation of the SIS-model and friendship paradox. Sec. \ref{sec:effects_step1} explores the effect of the distribution that samples agent $m$ in the first step of the SIS-model and show the invariance of the critical thresholds and states the first main result. Sec. \ref{sec:critical_thersholds} studies the effect of monphilic contagion, friendship paradox and degree assortativity on the contagion and states the second main result. Finally, Sec. \ref{sec:active_networks} shows how the collective dynamics of the SIS process on a random graph can be approximated by an ODE whose trajectory satisfies an algebraic constraint at every time instant.

\section{Preliminaries: Approximation of SIS Model and Friendship Paradox}
\label{sec:preliminaries}

In this section, the basic SIS model and how it can be approximated using deterministic mean-field dynamics is reviewed briefly. This approximation is utilized in the subsequent sections to obtain the main results. For detailed discussions about these results, the reader is encouraged to refer \cite{benaim2003deterministic, krishnamurthy2016, krishnamurthy2017tracking}.

\subsection{Discrete time SIS Model}
\label{subsec:SIS_model}

Consider a social network represented by an undirected graph $G = (V, E)$ where $V = \{1, 2, \dots, M\}$.  At each discrete time instant $n$, a node $v\in V$ of the network can take the state $s_n^{(v)} \in \{0,1 \}$ where, $0$ denotes the susceptible state and $1$ denotes the infected state. The degree $d(v) \in \{1,\dots, D\}$ of a node $v \in V$  is the number of nodes connected to $v$ and, $M(k)$ denotes the total number of nodes with degree $k$. Then, the degree distribution $P(k) = \frac{M(k)}{M}$ is the probability that a randomly selected node has degree $k$. Further, we also define the population state $\bar{x}_n(k)$ as the fraction of nodes with degree $k$ that are infected (state $1$) at time $n$ i.e.
\begin{equation}
\label{eq:population_state}
\bar{x}_n(k) = \frac{1}{M(k)}{\sum_{v \in V}\mathds{1}_{\{d(v) = k,\,s_n^{(v)} = 1 \}}}, \quad k = 1,\dots ,D.
\end{equation}

For this setting, we adopt the SIS model used in \cite{lopez2008diffusion, krishnamurthy2017tracking} which is as follows briefly. 

\vspace{0.1cm}
\noindent
{\bf Discrete Time SIS Model: }At each discrete time instant $n$,
\begin{compactenum}
	\item[\bf Step 1:] A node $m \in V$ is chosen with uniform probability $p^{X} (m) = 1/M$ where, $M$ is the number of nodes in the graph. 
	
	\item[\bf Step 2:] The state $s_n^{(v)} \in \{0,1 \}$ of the sampled node $m$ (in Step 1) evolves to $s_{n+1}^{(v)} \in \{0,1 \}$ with transition probabilities that depend on the degree of $m$, number of infected neighbors of $m$, population state of the network $\bar{x}_n$\footnote{$\bar{x}_n(k)$ is the fraction of infected nodes with degree $k$ i.e. $\bar{x}_n(k) = \frac{M^1(k)}{M(k)}$ where $M^1(k)$ is the number of infected nodes with degree $k$ and $M(k)$ is the number of nodes with degree $k$.} and the current state of $s^{(m)}_n$. 
\end{compactenum}

\vspace{0.25cm}
Note that the above model is a Markov chain with a state space consisting of $2^M$ states (since each of the $M$ nodes can be either infected or susceptible at any time instant). Due to this exponentially large state space, the discrete time SIS model is not mathematically tractable. However, we are interested only in the fraction of the infected nodes (as opposed to the exact state out of the $2^M$ states) and therefore, it is sufficient to focus on the dynamics of the population state $\bar{x}_n$ defined in (\ref{eq:population_state}) instead of the exact state of the infection.

\subsection{Deterministic Approximation by Mean-Field Dynamics}

The following result from \cite{krishnamurthy2017tracking} provides a useful means for obtaining a tractable deterministic approximation of the population state $\bar{x}_n$.

\begin{theorem}[Mean-Field Dynamics]
	\label{th:MFD}
	\begin{compactenum}
		\item The population state defined in (\ref{eq:population_state}) evolves according to the following stochastic difference equation driven by martingale difference process:      
		\begin{align}
			\label{eq:martingale_representation}
			\bar{x}_{n+1}(k) &= \bar{x}_{n}(k) + \frac{1}{M}[P_{01}(k, \bar{x}_n) - P_{10}(k, \bar{x}_n)] +\zeta_n \\
			\text{where,}\nonumber\\
			P_{01}(k, \bar{x}_n) &= (1-\bar{x}_n(k)) \times \nonumber\\
			&\hspace{1cm}\mathbb{P}(s_{n+1}^{m} = 1 \vert s_{n}^{m} = 0, d(m) = k, \bar{x}_n ) \label{eq:scaled_transition_prob_1}\\
			P_{10}(k, \bar{x}_n)  &= \bar{x}_n(k)\mathbb{P}(s_{n+1}^{m} = 0 \vert s_{n}^{m} = 1, d(m) = k, \bar{x}_n ).\label{eq:scaled_transition_prob_2}
		\end{align} are the scaled transition probabilities of the states and, $\zeta_n$ is a martingale difference process with $\vert \vert \zeta_n\vert \vert_2 \leq \frac{\Gamma}{M}$ for some positive constant $\Gamma$.

		\item Consider the mean-field dynamics process associated with the population state:
		\begin{equation}
		\label{eq:MFD_approximation_X}
		x_{n+1} (k) = x_{n}(k) + \frac{1}{M}\big(	P_{01}(k, x_n) - P_{10}(k, x_n)	\big)
		\end{equation}
		where, $P_{01}(k, x_n)$ and $P_{10}(k, x_n)$ are as defined in (\ref{eq:scaled_transition_prob_1}), (\ref{eq:scaled_transition_prob_2}) and $x_0 = \bar{x}_0$.
		Then, for a time horizon of T points, the deviation between the mean-field dynamics (\ref{eq:MFD_approximation_X})  and the actual population state $\bar{x}_n$ of the SIS model satisfies
		\begin{equation}
		\mathbb{P}\{\underset{0\leq n \leq T}{\max}\vert\vert  x_n -\bar{x}_n   \vert\vert_{\infty} \geq \epsilon   \} \leq C_1\exp(C_2 \epsilon^2M)
		\end{equation} for some positive constants $C_1, C_2$ providing $T=O(M)$.
	\end{compactenum}
\end{theorem}

First part of Theorem \ref{th:MFD} is the martingale representation of a Markov chain (which is the population state $\bar{x}_n$). Note from (\ref{eq:martingale_representation}) that the dynamics of the population state $\bar{x}_n$ resemble a stochastic approximation recursion (new state is the old state plus a noisy term). Hence, the trajectory of the population state $\bar{x}_n$ should converge (weakly) to the deterministic trajectory given by the ODE corresponding to the mean-field dynamics in (\ref{eq:MFD_approximation_X}) as the size of the network $M$ goes to infinity i.e. the step size of the stochastic approximation algorithm goes to zero (for details, see \cite{krishnamurthy2016, kushner2003}). Second part of the theorem provides an exponential bound on the deviation of the mean-field dynamics approximation from the actual population state for a finite length of the sample path. In the subsequent sections of this paper, the mean-field approximation (\ref{eq:MFD_approximation_X}) is utilized to study the effects of various sampling methods and friendship paradox on the SIS model of information diffusion.

\subsection{Friendship Paradox}
\label{subsec:friendship_paradox}

\textit{Friendship paradox} refers to a graph theoretic consequence that was introduced in 1991 by Scott. L. Feld in \cite{feld1991}. We briefly review of the main results related to friendship paradox in this subsection. 
Feld's original statement of the friendship paradox is ``on average, the number of friends of a random friend is always greater than or equal to the number of friends of a random individual". Here, a random friend refers to a random end node $Y$ of a randomly chosen edge (a pair of friends). This statement is formally stated in Theorem \ref{th:friendship_paradox}. Further, Theorem \ref{th:friendship_paradox_v2_cao} (based on \cite{cao2016}) states that a similar result holds when the degrees of a random node $X$ and random friend $Z$ of a random node $X$ are compared as well. 

\begin{theorem}[Friendship Paradox - Version 1 \cite{feld1991}]
	\label{th:friendship_paradox}
	Let ${G = (V,E)}$ be an undirected graph, $X$ be a node chosen uniformly from $V$ and, $Y$ be a uniformly chosen node from a uniformly chosen edge $e\in E$. Then,
	\begin{equation}
	\mathbb{E} \{d(Y)\} \geq \mathbb{E}\{d(X)\},
	\end{equation} where, $d(X)$ denotes the degree of $X$. 
	
\end{theorem}

\begin{theorem}[Friendship Paradox - Version 2 \cite{cao2016}]
	\label{th:friendship_paradox_v2_cao}
	Let ${G = (V,E)}$ be an undirected graph, $X$ be a node chosen uniformly from $V$ and, $Z$ be a uniformly chosen neighbor of a uniformly  chosen node from $V$. Then, 
	\begin{equation}
	\label{eq:fosd}
	d(Z)\geq_{fosd} d(X)
	\end{equation} 
	where, $\geq_{fosd}$ denotes the first order stochastic dominance\footnote{\label{fn:fosd}A discrete random variable $X$ (with a cumulative distribution function $F_X$) first order stochastically dominates a discrete random variable $Y$ (with a cumulative distribution function $F_Y$), denoted $X \geq_{fosd} Y$ if, ${F_X(n) \leq F_Y(n)}$, for all $n$. Further, first order stochastic dominance implies larger mean.}. 
\end{theorem}

The intuition behind Theorem \ref{th:friendship_paradox} and Theorem \ref{th:friendship_paradox_v2_cao} stems from the fact that individuals with a large number of friends (high degree nodes) appear as the friends of a large number of individuals. Hence, these high degree nodes can contribute to an increase in the average number of friends of friends. On the other hand, individuals with smaller number of friends appear as friends of a smaller number of individuals. Hence, they do not cause a significant change in the average number of friends of friends. 

Friendship paradox, which in essence is a sampling bias observed in undirected social networks has gained attention as a useful tool for estimation and detection problems in social networks. For example, \cite{eom2015tail} proposes to utilize friendship paradox as a sampling method for reduced variance estimation of a heavy-tailed degree distribution, \cite{christakis2010,garcia2014using, sun2014efficient} explore how the friendship paradox can be used for detecting a contagious outbreak quickly, \cite{seeman2013, lattanzi2015,horel2015scalable, kim2015social, kumar2018network} utilizes friendship paradox for maximizing influence in a social network, \cite{nettasinghe2018your} proposes friendship paradox based algorithms for efficiently polling a social network (e.g. to forecast an election) in a social network, \cite{jackson2016friendship} studies how the friendship paradox in a game theoretic setting can systematically bias the individual perceptions. Further, \cite{eom2014,hodas2013,higham2018centrality,lerman2016,bagrow2017friends, kramer2016multistep, bollen2017happiness, fotouhi2014generalized} present and analyze generalizations of the classical friendship paradox other attributes and networks. 

\section{Effect of the Sampling Distribution in the Step 1 of the SIS Model }
\label{sec:effects_step1}
Recall from Sec. \ref{subsec:friendship_paradox} that we distinguished between three sampling methods for a network ${G = (V, E)}$: a random node $X$, a random friend $Y$ and, a random friend $Z$ of a random node. Further, recall that in the discrete-time SIS model explained in Sec. \ref{subsec:SIS_model}, the node $m$ that whose state evolves is sampled uniformly from $V$ i.e. $m\overset{d}{=} X$. This section studies the effect of random friends ($Y$ or $Z$) evolving at each time instant instead of random nodes ($X$) i.e. the cases where ${m\overset{d}{=} Y}$ or ${m\overset{d}{=} Z}$. Following is our main result in this section:
\begin{theorem}
	\label{thm:MFD_samping_YZ}
	Consider the discrete time SIS model explained previously.
	\begin{compactenum}
		\item If the node $m$ is a random end  $Y$ of random link i.e. node $m$ with degree $d(m)$ is chosen with probability ${p^Y(m) = \frac{d(m)}{\sum_{v\in V}d(v)}}$, then the stochastic dynamics of the SIS model can be approximated by,
		\begin{equation}
		\label{eq:MFD_sampling_Y}
		x_{n+1} (k) = x_{n}(k) + \frac{1}{M} \frac{k}{\bar{k}}\big(	P_{01}(k, x_n) - P_{10}(k, x_n)	\big),
		\end{equation}
		where $\bar{k}$ is the average degree of the graph $G = (V, E)$. 
		
		\item If the node $m$ is a random neighbor $Z$ of a random node $X$, then the stochastic dynamics of the SIS model can be approximated by,
		\begin{align}
			x_{n+1} (k) &= x_{n}(k) +\frac{1}{M} \bigg({\sum_{k'}\frac{P(k)}{P(k')}P(k\vert k')}\bigg)\times \nonumber\\
			&\hspace{2cm} \big(	P_{01}(k, x_n) - P_{10}(k, x_n)	\big), 	\label{eq:MFD_sampling_Z}
		\end{align}
		where $\bar{k}$ is the average degree of the graph $G = (V, E)$, $P$ is the degree distribution and $P(k\vert k')$ is the probability that a random neighbor of a degree $k'$ node is of degree $k$.  Further, if the network is a degree-uncorrelated network i.e. $P(k\vert k')$ does not depend on $k'$, then (\ref{eq:MFD_sampling_Z}) will be the same as $(\ref{eq:MFD_sampling_Y})$. 
	\end{compactenum}
\end{theorem}

\begin{proof}
Note that the population state $\{\bar{x}_{n}\}_{n_\geq 0 }$ is a Markov chain with a state space of the size ${\Pi_{d = 1}^M (M(d) + 1)}$. Let $P^{pop}$ denote the transition probability matrix of this Markov chain and $e_i$ denote the ${\Pi_{d = 1}^M (M(d) + 1)}$ dimensional column vector with $1$ in the $i^{th}$ position and zeros in all other positions. Then, the Martingale representation of this Markov chain is,
\begin{equation}
\label{eq:martingale_rep_1}
\rho_{n+1} = (P^{pop})' \rho_n + \eta_n 
\end{equation}
where, $\rho_n$ are states taking values in the space $\{e_1,\dots,e_{\Pi_{d = 1}^M (M(d) + 1)}\}$,
$\eta_n$ is martingale difference noise. Then, by multiplying with the state level matrix, we get
\begin{align}
\bar{x}_{n+1}(k) = \mathbb{E}\{\bar{x}_n(k)\vert \bar{x}_n\} + \gamma_n
\end{align}
where, $\gamma_n$ is the product of martingale difference noise $\eta_n$ and state level matrix. Then,
\begin{align}
&\bar{x}_{n+1}(k) = \mathbb{E}\{\bar{x}_n(k)\vert \bar{x}_n\} + \gamma_n\\
&= \mathbb{P}(s_{n+1}^{(m)} = 1, s_{n}^{(m)} = 0, d(m) = k \vert \bar{x}_n) \times (\bar{x}_n(k) + \frac{1}{M(k)}) +\nonumber\\ &\hspace{0.45cm}\mathbb{P}(s_{n+1}^{(m)} = 0, s_n^{(m)} = 1, d(m) = k \vert \bar{x}_n) \times (\bar{x}_n(k) - \frac{1}{M(k)}) +\nonumber\\
& \hspace{1.5cm}(1 - \mathbb{P}(s_{n+1}^{(m)} = 1, s_n^{(m)} = 0, d(m) = k \vert \bar{x}_n) - \nonumber\nonumber\\
& \hspace{0.45cm}\mathbb{P}(s_{n+1}^{(m)} = 0, s_{n}^{(m)} = 1, d(m) = k \vert \bar{x}_n)) (\bar{x}_n(k))
+ \gamma_n
\end{align}
Let,
\begin{align}
A &= \mathbb{P}(s_{n+1}^{(m)} = 1, s_{n}^{(m)} = 0, d(m) = k \vert \bar{x}_n) \nonumber\\
B &= \mathbb{P}(s_{n+1}^{(m)} = 0, s_{n}^{(m)} = 1, d(m) = k \vert \bar{x}_n) \nonumber.
\end{align}
Then, we get
\begin{align}
\bar{x}_{n+1}(k) &= \bar{x}_{n}(k) + \frac{1}{M(k)}(A - B) + \gamma_n.
\end{align}
Then, the first and second parts of the Theorem \ref{thm:MFD_samping_YZ} follow by decomposing the joint distributions of $A, B$ with respective to the degree distributions of a random friend $Y$ and a random friend $Z$ of a random node $X$ respectively. 

\end{proof}

Theorem \ref{thm:MFD_samping_YZ} shows that, if the node $m$ sampled in the step 1 of the SIS model (explained in Sec. \ref{subsec:SIS_model}), is chosen to be a random friend or a random friend of a random node, then different elements $x_n(k)$ of the mean-field approximation evolves at different rates. This result allows us to model the dynamics of the population state in the more involved case where, frequency of the evolution of an individual is proportional his/her degree (part 1 - e.g. high degree nodes change opinions more frequently due to higher exposure) and also depends on the degree correlation (part 2 - e.g. nodes being connected to other similar/different degree nodes changes the frequency of changing the opinion).

\vspace{0.25cm}
\begin{remark} [Invariance of the critical thresholds to the sampling distribution in step 1] 
	\normalfont
	The stationary condition for the mean-field dynamics is obtained by setting $x_{n+1}(k) - x_{n}(k) = 0$ for all $k \geq 1$. Comparing (\ref{eq:MFD_approximation_X}) with (\ref{eq:MFD_sampling_Y}) and (\ref{eq:MFD_sampling_Z}), it can be seen that this condition yields the same expression $P_{01}(k, x_n) - P_{10}(k, x_n) = 0$, for all three sampling methods (random node - $X$, random end of a random link $Y$ and, a random neighbor $Z$ of a random node). Hence, the critical thresholds of the SIS model are invariant to the distribution from which the node $m$ is sampled in step 1. This leads us to Sec. \ref{sec:critical_thersholds} where, modifications to the step 2 of the SIS model are analyzed in terms of the critical thresholds. 
\end{remark}

\section{Critical Thresholds for Unbiased-degree Networks}
\label{sec:critical_thersholds}
In Sec. \ref{sec:effects_step1} of this paper, we focused on the step 1 of the SIS model and, showed that different sampling methods for selecting the node $m$ result in different mean-field dynamics with the same stationary conditions. In contrast, the focus of this section is on the step 2 of the SIS model and, how changes to this step would result in different stationary conditions and critical thresholds. 
\subsection{Critical Thresholds for Monophilic and Non-Monophilic Contagions}

Recall the SIS model reviewed in Sec. \ref{subsec:SIS_model} again. We limit our attention to the case of \textit{unbiased-degree} networks and viral adoption rules discussed in \cite{lopez2016overview}. 

\vspace{0.25cm}
\noindent
{\bf Unbiased-degree network:} In an unbiased-degree network, neighbors of agent $m$ sampled in the step 1 of the SIS model are $d(m)$ (degree of agent $m$) number of uniformly sampled agents (similar in distribution to the random variable $X$) from the network. Therefore, in an unbiased-degree network, any agent is equally likely to be a neighbor of the sampled (in the step 1 of the SIS model) agent $m$.

\vspace{0.25cm}
\noindent
{\bf Viral adoption rules\footnote{	The above two rules are called viral adoption rules as they consider the total number of infected nodes (denoted by $a^X_m$ and $a^Z_m$ in case 1 and case 2 respectively) in the sample in contrast to the persuasive adoption rules that consider the fraction of infected nodes in the sample \cite{lopez2012influence}.}:}  If the sampled agent $m$ (in the step 1 of the SIS model) is an infected agent, she becomes susceptible with a constant probability $\delta$.  If the sampled agent $m$ (in the step 1 of the SIS model) is a susceptible (state $0$) agent, she samples $d(m)$ (degree of $m$) number of other agents $X_1, X_2, \dots, X_{d(m)}$ (neighbors of $m$ in the unbiased-degree network) from the network and, adopts the contagion based on one of the following rules:
\begin{compactenum}
	\item [ Case 1 - Non-monophilic adoption rule:] For each sampled neighbor $X_i$, $m$ observes the state of $X_i$. Hence, agent $m$ observes the states of $d(m)$ number of random nodes. Let $a^X_m$ denote the number of infected agents among $X_1,\dots, X_{d(m)}$. Then, the susceptible agent $m$ becomes infected with probability $\nu\frac{a^X_m}{D}$ where, $0 \leq \nu\leq 1$ is a constant and $D$ is the largest degree of the network.
	
	\vspace{0.2cm}
	\item[ Case 2 - Monophilic adoption rule:] For each sampled neighbor $X_i$, $m$ observes the state of a random friend $Z_i\in \mathcal{N}(X_i)$ of that neighbor. Hence, agent $m$ observes the states of $d(m)$ number of random friends $Z_1,\dots, Z_{d(m)}$ of random nodes $X_1,\dots,X_{d(m)}$. Let $a^Z_m$ be the number of infected agents among $Z_1, \dots, Z_{d(m)}$. Then, the susceptible agent $m$ becomes infected with probability $\nu\frac{a^Z_m}{D}$ where, $0 \leq \nu\leq 1$ is a constant and $D$ is the largest degree of the network.
\end{compactenum}

In order to compare the effects of non-monophilic and monophilic adoption rules, we look at the conditions on the model parameters for which, each rule leads to a positive fraction of infected nodes starting from a small fraction of infected nodes i.e. a positive stationary solution to the mean-field dynamics (\ref{eq:MFD_approximation_X}). Our main result is the following:
\begin{theorem}
	\label{thm:observe_XZ}
	Consider the SIS model described in Sec. \ref{subsec:SIS_model}. Define the effective spreading rate as $\lambda = \frac{\nu}{\delta}$ and let $X$ be a random node and $Z$ be a random friend of $X$.
	\begin{compactenum}
		\item Under the non-monophilic adoption rule (Case 1), the mean-field dynamics equation ($\ref{eq:MFD_approximation_X}$) takes the form,
		\begin{align}
			\label{eq:MFD_approximation_X_observe_X}
			x_{n+1} (k) &= x_{n}(k) + \frac{1}{M}\big(	(1-{x}_n(k))\frac{\nu k \theta_n^X}{D}- x_n(k)\delta\big)\\
			\quad \text{where},\nonumber\\
			\theta_n^X &= \sum_{k}P(k) x_n(k) 
		\end{align} is the probability that a randomly chosen node $X$ at time $n$ is infected. Further, there exists a positive stationary solution to the mean field dynamics (\ref{eq:MFD_approximation_X_observe_X}) for case 1 if and only if
		\begin{equation}
		\label{eq:diff_threshold_X}
		\lambda > \frac{D}{\mathbb{E}\{d(X)\}} = \lambda^*_X
		\end{equation}
		
		\item Under the monophilic adoption rule (Case 2), the mean-field dynamics equation (\ref{eq:MFD_approximation_X}) takes the form,
		\begin{align}
			\label{eq:MFD_approximation_X_observe_Z}
			x_{n+1} (k) &= x_{n}(k) + \frac{1}{M}\big(	(1-{x}_n(k))\frac{\nu k \theta_n^Z}{D}- x_n(k)\delta	\big)\\
			\text{where,}\nonumber\\
			\theta_n^Z &= \sum_{k}\bigg(\sum_{k'}P(k')P(k|k')\bigg) x_n(k)
		\end{align} is the probability that a randomly chosen friend $Z$ of a randomly chosen node $X$ at time $n$ is infected\footnote{\label{fn:jdd}We use $P(k|k')$ to denote the conditional probability that a node with degree $k'$ is connected to a node with degree $k$. More specifically $P(k|k') = \frac{e(k,k')}{q(k)}$ where $e(k,k')$ is the joint degree distribution of the network and $q(k)$ is the marginal distribution that gives the probability of random end (denoted by random variable $Y$ in Theorem \ref{th:friendship_paradox}) of random link having degree $k$. We also use $\sigma_q$ to denote the variance of $q(k)$ in subsequent sections.}. Further, there exists a positive stationary solution to the mean field dynamics (\ref{eq:MFD_approximation_X_observe_Z}) if and only if
		\begin{equation}
		\label{eq:diff_threshold_Z}
		\lambda > \frac{D}{\mathbb{E}\{d(Z)\}}= \lambda^*_Z
		\end{equation}
	\end{compactenum}
\end{theorem}

\begin{proof}
\noindent
{\bf Part 1: Non-monophilic adoption rule:} The proof of the first part is inspired by \cite{lopez2016overview,lopez2012influence} that consider the unbiased degree networks with non-monophilic adoption rules with continuous-time evolutions (as opposed to the discrete time case considered here). The main purpose of the first part is to provide a comparison of the non-monophilic adoption rule with the monophilic adoption rule (part 2). Consider the mean-field dynamics given in \ref{eq:MFD_approximation_X}. The probability of a susceptible agent agent $m$ (with degree $d(m) = k$) sampled at time instant $n$ for the unbiased degree network adopting the contagion can be derived as follows:
\begin{align}
\label{eq_proof:transition01_prob_binom_expectation}
&\mathbb{P}(s_{n+1}^{m} = 1 \vert s_{n}^{m} = 0, d(m) = k, \bar{x}_n ) =\\
&\hspace{1cm} \sum_{a = 1}^{k} \frac{\nu a}{D} \binom{k}{a} (\theta_n^X)^a (1 - \theta_n^X)^{(k-a)} =\frac{\nu k \theta_n^X}{D}
\end{align}
where,
\begin{equation}
\theta_n^X = \sum_{k}P(k) x_n(k) 
\end{equation} 
is the probability that a randomly chosen node $X$ is infected at time instant $n$. Eq.	(\ref{eq_proof:transition01_prob_binom_expectation}) is based on the following argument. The neighbors of $m$ (in the case of non-monophilic adoption rule) are are randomly sampled nodes ($X$) and therefore, the number $a$ of infected neighbors (out of $k$ total) follows a binomial distribution with parameter $\theta_n^X$. Since the probability of being infected when a susceptible node $m$ has $a$ (out of $k$) infected neighbors is $\frac{\nu a}{D}$, the probability of a degree $k$ susceptible node becoming infected at time $m$ is the expectation of $\frac{\nu a}{D}$ with respect to the binomial distribution as calculated in $(\ref{eq_proof:transition01_prob_binom_expectation})$.
Further, the probability that an infected node (independent of the degree) becomes susceptible is $\delta$ as assumed in the viral adoption rule i.e.
\begin{equation}
\label{eq_proof:transition10}
\mathbb{P}(s_{n+1}^{m} = 0 \vert s_{n}^{m} = 1, d(m) = k, \bar{x}_n ) = \delta.
\end{equation}
Then, substituting (\ref{eq_proof:transition10}), (\ref{eq_proof:transition01_prob_binom_expectation}) in to the mean-field dynamics \ref{eq:MFD_approximation_X} yields the mean-field dynamics approximation (\ref{eq:MFD_approximation_X_observe_X}) for non-monophilic adoption rules.

In order to obtain the critical thresholds for non-monophilic case, consider the stationary condition of \ref{eq:MFD_approximation_X_observe_X}, $x_{n+1}(k) - x_{n}(k) = 0$, which yields the stationary population state $x(k)$ to follow,
\begin{align}
&x(k) 
%&= \frac{\mathbb{P}(s_{n+1}^{m} = 1 \vert s_{n}^{m} = 0, d(m) = k, \bar{x} )}{\mathbb{P}(s_{n+1}^{m} = 1 \vert s_{n}^{m} = 0, d(m) = k, \bar{x}_n ) + \mathbb{P}(s_{n+1}^{m} = 0 \vert s_{n}^{m} = 1, d(m) = k, \bar{x} )}\\
= \frac{\frac{\nu k \theta^X}{D}}{\frac{\nu k \theta^X}{D} + \delta}
\label{eq_proof:stationary_pop_state_non_mono}
=\frac{\lambda k \theta^X}{\lambda k \theta^X + D}\\
\noindent \text{where,}\nonumber\\ 
\label{eq_proof:stationary_infected_fraction}
&\theta^X = \sum_{k}P(k)x(k)= \rho 
\end{align}
is the probability that a randomly chosen node $X$ is infected during stationary state (defined in \ref{eq:stationary_infected_fraction}). Then, by substituting (\ref{eq_proof:stationary_pop_state_non_mono}) in (\ref{eq_proof:stationary_infected_fraction}), we get
\begin{align}
\label{eq_proof:H_expression_non_mono}
\rho = \sum_{k}P(k)\frac{\lambda k \rho}{\lambda k \rho + D} =H^{X}_{\lambda, P}(\rho).
\end{align}

The diffusion prevails (without dying away) when (\ref{eq_proof:H_expression_non_mono}) has a positive solution. Further, (\ref{eq_proof:H_expression_non_mono}) is an increasing, concave function with $H^{X}_{\lambda, P}(0) = 0$ and
$H^{X}_{\lambda, P}(1) < 1$.

Hence, in order for the (\ref{eq_proof:H_expression_non_mono}) to have a positive solution,
\begin{align}
\frac{dH^{X}_{\lambda, P}(\rho)}{d\rho}\Bigg\vert_{\rho = 0} &> 1,
\end{align}
which then yields,
\begin{align}
\lambda &> \frac{D}{\sum_{k}kP(k)} = \frac{D}{\mathbb{E}d(X)}
\end{align}

\noindent
{\bf Part 2: Monophilic Adoption Rule:} Consider the mean-field dynamics given in \ref{eq:MFD_approximation_X}. By following steps similar to the part 1 of the proof, The probability of a susceptible agent agent $m$ (with degree $d(m) = k$) sampled at time instant $n$ for the unbiased degree network adopting the contagion is:
\begin{align}
\label{eq_proof:transition01_prob_binom_expectation_Monophilic}
\mathbb{P}(s_{n+1}^{m} = 1 \vert s_{n}^{m} = 0, d(m) = k, \bar{x}_n ) &=\frac{\nu k \theta_n^Z}{D}
\end{align}
where,
\begin{equation}
\theta_n^Z = \sum_{k}\bigg(\sum_{k'}P(k')P(k|k')\bigg) x_n(k)
\end{equation} 
is the probability that a randomly chosen friend $Z$ of a randomly chosen node $X$ is infected at time instant $n$. To understand how the expression for $\theta_n^Z$ is derived, recall that $P(k|k')$ is the probability that a randomly chosen friend of a node with degree $k'$ is degree $k$. Then, taking expectation of $P(k|k')$ with respect to probability of sampling a node with degree $k'$ yields $\sum_{k'}P(k')P(k|k')$ as the probability of a random friend $Z$ (of a random node $X$) having a degree $k$.

Further, the probability that an infected node (independent of the degree) becomes susceptible is $\delta$ as assumed in the viral adoption rule i.e.
\begin{equation}
\label{eq_proof:transition10_mono}
\mathbb{P}(s_{n+1}^{m} = 0 \vert s_{n}^{m} = 1, d(m) = k, \bar{x}_n ) = \delta.
\end{equation}
Then, substituting (\ref{eq_proof:transition10_mono}), (\ref{eq_proof:transition01_prob_binom_expectation_Monophilic}) in to the mean-field dynamics \ref{eq:MFD_approximation_X} yields the mean-field dynamics approximation (\ref{eq:MFD_approximation_X_observe_Z}) for non-monophilic adoption rules.

In order to obtain the critical thresholds for non-monophilic case, consider the stationary condition of \ref{eq:MFD_approximation_X_observe_Z}, $x_{n+1}(k) - x_{n}(k) = 0$, which yields the stationary population state $x(k)$ to follow,
\begin{align}
x(k) 
%&= \frac{\mathbb{P}(s_{n+1}^{m} = 1 \vert s_{n}^{m} = 0, d(m) = k, \bar{x} )}{\mathbb{P}(s_{n+1}^{m} = 1 \vert s_{n}^{m} = 0, d(m) = k, \bar{x}_n ) + \mathbb{P}(s_{n+1}^{m} = 0 \vert s_{n}^{m} = 1, d(m) = k, \bar{x} )}\\
&= \frac{\frac{\nu k \theta^Z}{D}}{\frac{\nu k \theta^Z}{D} + \delta}
\label{eq_proof:stationary_pop_state_mono}
=\frac{\lambda k \theta^Z}{\lambda k \theta^Z + D}\\
\noindent \text{where,}\nonumber \\ 
\label{eq_proof:stationary_infected_fraction_mono}
\theta^Z &= \sum_{k}\bigg(\sum_{k'}P(k')P(k|k')\bigg) x(k)
\end{align}
is the probability that a random friend $Z$ of a randomly node $X$ is infected during stationary state (defined in \ref{eq:stationary_infected_fraction}). Then, by substituting (\ref{eq_proof:stationary_pop_state_mono}) in (\ref{eq_proof:stationary_infected_fraction_mono}), we get
\begin{align}
\label{eq_proof:H_expression_mono}
\theta^Z &= \sum_{k}\bigg(\sum_{k'}P(k')P(k|k')\bigg)\frac{\lambda k \theta^Z}{\lambda k \theta^Z + D} =H^{Z}_{\lambda, P}(\theta^Z).
\end{align}

The diffusion prevails (without dying away) when (\ref{eq_proof:H_expression_mono}) has a positive solution. Further, (\ref{eq_proof:H_expression_mono}) is an increasing, concave function of $\theta^Z$ with
$H^{Z}_{\lambda, P}(0) = 0$ and $H^{Z}_{\lambda, P}(1)  < 1$. 

Hence, in order for the (\ref{eq_proof:H_expression_mono}) to have a positive solution,
\begin{align}
\frac{dH^{Z}_{\lambda, P}(\theta^Z)}{d\theta^Z}\Bigg\vert_{\theta^Z = 0} &> 1,
\end{align}
which then yields,
\begin{align}
\lambda &> \frac{D}{\sum_{k}k(\sum_{k'}P(k')P(k|k')) }= \frac{D}{\mathbb{E}d(Z)}.
\end{align}
\end{proof}
The infection spreading under the monophilic adoption rule (Case 2 of Theorem \ref{thm:observe_XZ}) can also be thought of as representing the network by the square graph (corresponding to the square of the adjacency matrix of the original network). Proceeding that way would also yield the same critical threshold as in the Case 2 of Theorem \ref{thm:observe_XZ}. Theorem \ref{thm:observe_XZ} allows us to analyze the effects of friendship paradox and degree-assortativity on the contagion process as discussed in the next subsection.

\subsection{Effects of Friendship Paradox and Degree Correlation on the Monophilic Contagion}

\begin{figure*}[]
	\centering
	\begin{subfigure}[!h]{0.45\textwidth}
		\centering
		\includegraphics[height = 3.5in, width = 3.6in]{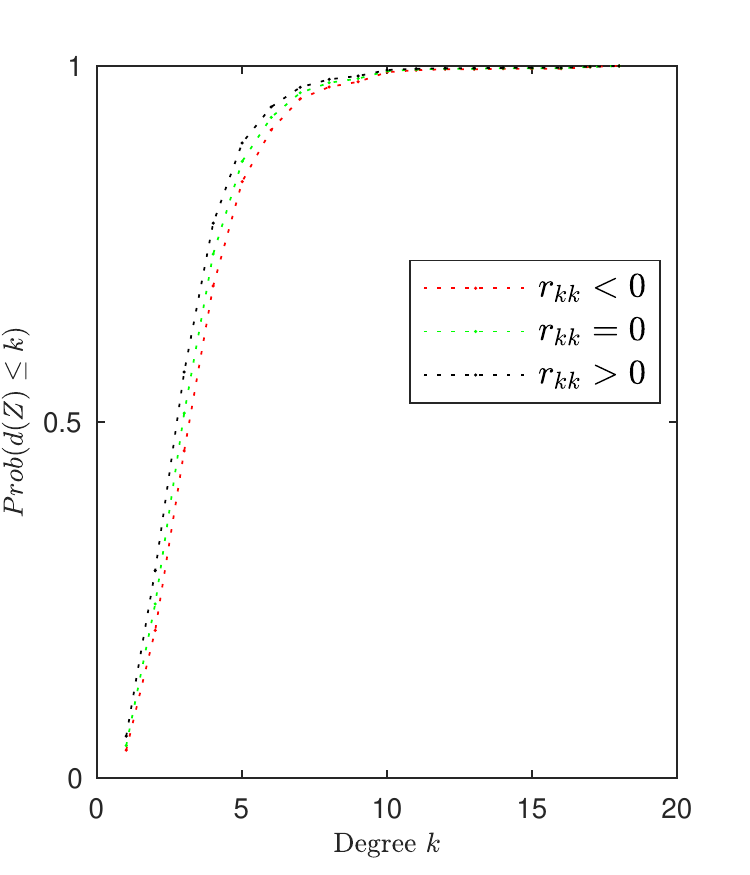}
		\caption{CDFs of the of the degree $d(Z)$  of a random friend $Z$ of a random node for three networks with same degree distribution but different assortativity $r_{kk}$ values. Note that the CDFs are point-wise increasing with $r_{kk}$ showing that $\mathbb{E}\{d(Z)\}$ decreases with $r_{kk}$.}
		\label{subfig:cdf_degree_z_with_rkk}
	\end{subfigure}\hfill
	\begin{subfigure}[!h]{0.45\textwidth}
		\centering
		\includegraphics[height = 3.5in, width = 3.5in]{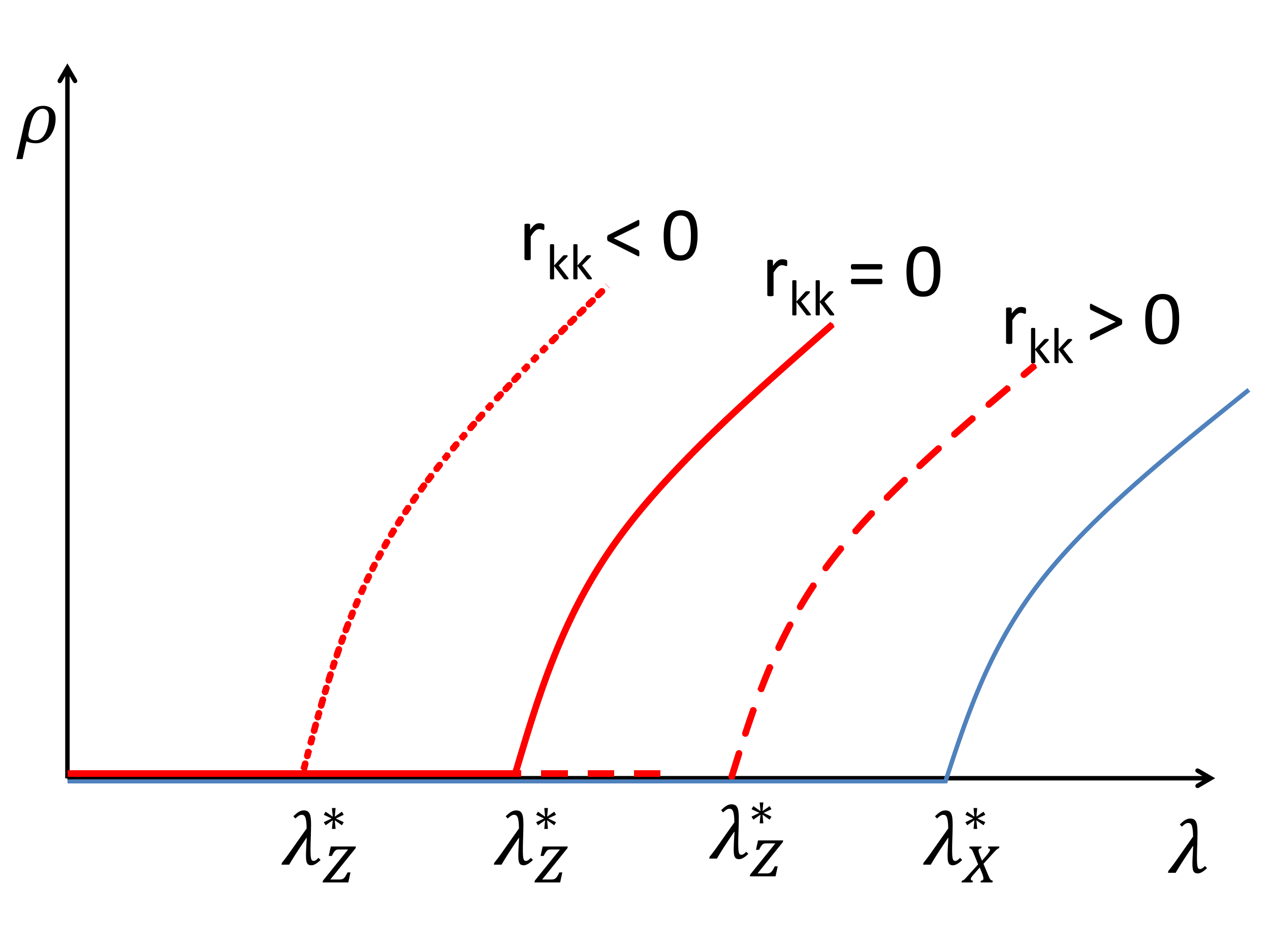}
		\caption{Variation of the stationary fraction $\rho$ of infected nodes with the effective spreading rate $\lambda$ for the case 1 (blue) and case 2 (red), illustrating the ordering of the critical thresholds of cases 1,2 and the effect of assortativity.}
		\label{subfig:stationary_infected_fraction}
	\end{subfigure}\hfill
	\caption{Comparison of non-monophilic and monophilic adoption rules and the effect of assortativity on the critical thresholds of the monophilic adoption rule.}
	\label{fig:assortativity_effect}
\end{figure*}

Theorem \ref{thm:observe_XZ} showed that the critical thresholds of the mean-filed dynamics equation (\ref{eq:MFD_approximation_X}) for the two adoption rules (non-monophilic and monophilic contagion) are different. Following is an immediate corollary of Theorem \ref{thm:observe_XZ} which gives the ordering of these critical thresholds using the friendship paradox stated in Theorem \ref{th:friendship_paradox}. 
\begin{corollary}
	\label{cor:ordering_crtical_thresholds}
	The critical thresholds $\lambda^*_X, \lambda^*_Z$ in (\ref{eq:diff_threshold_X}), (\ref{eq:diff_threshold_Z}) for the cases of non-monophilic (case 1) and monophilic (case 2) adoption rules satisfy
	\begin{equation}
	\lambda^*_Z \leq \lambda^*_X.
	\end{equation}
\end{corollary}

Corollary \ref{cor:ordering_crtical_thresholds} shows that in the case of monophilic adoption rule, it is easier (smaller effective spreading rate) for the infection to spread to a positive fraction of the agents as a result of the friendship paradox. Hence, observing random friends of random neighbors for adopting a contagion makes it easier for the contagion to spread instead of dying away (in unbiased-degree networks). This shows how friendship paradox can affect the spreading of a contagion over a network. 

\begin{remark}
	If we interpret an individual's second-hop connections as weak-ties, then Theorem \ref{thm:observe_XZ} and Corollary \ref{cor:ordering_crtical_thresholds} can be interpreted as results showing the importance of weak-ties in contagions (in the context of a SIS model and an unbiased-degree network). See the seminal works in \cite{rapoport1953spread, granovetter1973strength} for the importance and definitions of weak-ties in the sociology context. 
\end{remark}

The ordering $\lambda^*_Z \leq \lambda^*_X$ of the critical thresholds in Corollary \ref{cor:ordering_crtical_thresholds} holds irrespective of any other network property. However, the magnitude of the difference of the critical thresholds $\lambda^*_X - \lambda^*_Z$ depends on the neighbor-degree correlation (assortativity) coefficient defined as,
\begin{equation}
\label{eq:deg_deg_corr}
r_{kk} = \frac{1}{\sigma_q^2}\sum_	{k,k'}kk'\Big(e(k,k')  - q(k)q(k')\Big)
\end{equation} using the notation defined in Footnote \ref{fn:jdd}. To intuitively understand this, consider a star graph that has a negative assortativity coefficient (as all low degree nodes are connected to the only high degree node). Therefore, a randomly chosen node $X$ from the star graph has a much smaller expected degree $\mathbb{E}\{d(X)\}$ than the expected degree $\mathbb{E}\{d(Z)\}$ of a random friend $Z$ of the random node $X$ compared to the case where the network has a positive assortativity coefficient. This phenomenon is further illustrated in Fig. \ref{fig:assortativity_effect} using three networks with the same degree distribution but different assortativity coefficients obtained using Newman's edge rewiring procedure  \cite{newman2002assortative}.  

Consider the stationary fraction of the infected nodes \begin{equation}
\label{eq:stationary_infected_fraction}
\rho = \sum_{k}P(k)x(k)
\end{equation} where $P(k)$ is the degree distribution and $x(k), k = 1,\dots, D$ are the stationary states of the mean-field dynamics in (\ref{eq:MFD_approximation_X}). Fig. \ref{subfig:stationary_infected_fraction} illustrates how the stationary fraction of the infected nodes varies with the effective spreading rate $\lambda$ for case 1 and 2, showing the difference between the two cases and the effect of assortativity.

\section{Collective Dynamics of SIS-Model and Reactive Networks}
\label{sec:active_networks}

So far in Sec. \ref{sec:effects_step1} and Sec. \ref{sec:critical_thersholds}, the underlying social network on which the contagion spreads was treated as a deterministic graph and, the mean-field dynamics equation (\ref{eq:MFD_approximation_X}) was used to approximate the SIS-model. In contrast, this section explores the more involved case where the underlying social network also randomly evolves at each time step $n$ (of the SIS-model) in a manner that depends on the population state $\bar{x}_n$. Our aim is to obtain a tractable model that represent the collective dynamics of the SIS-model and the evolving graph process. As explained in Sec. \ref{subsec:motivation} with examples, the motivation for this problem comes from the real world networks that evolves depending on the state of diffusions on them. In order to state the main result, we consider the following general setup.

\vspace{0.25cm}
\noindent
{\bf Reactive Network:} Consider a finite set $\mathcal{G} = \{\mathcal{G}_1,\dots,\mathcal{G}_N\}$  of $N$ graphs that has the same degree distribution $P(k)$ but different conditional degree distributions $P_{\mathcal{G}_1}(k|k'),\dots, P_{\mathcal{G}_N}(k|k')$. The Markovian graph process $\{G_n\}_{n\geq 0}$ with state space $\mathcal{G}$ evolves with the transition probabilities parameterized by the population state $\bar{x}_n$ i.e. $G_{n+1} \sim P_{\bar{x}_n}(\,\cdot\,|G_n)$ (assumed to be irreducible, positive recurrent with a unique stationary distribution $\pi_{\bar{x}_n}$) where, the parameterization by $\bar{x}_{n}$ represents the functional dependency of the graph process on the state of the SIS diffusion process i.e. a reactive network. For example, this may represent a network that (stochastically) changes its assortativity adversarially in order to prevent the diffusion, etc. as discussed in Sec. \ref{subsec:motivation}. 

In this context, our main result is the following.
\begin{theorem}[Collective Dynamics of SIS-model and Reactive Network]
	\label{th:active_network}
	Consider a reactive network (stated in Sec. \ref{sec:active_networks}) with transition probability matrix $P_{\bar{x}_n}(\,\cdot\,|G_n)$  (with stationary distribution $\pi_{\bar{x}_n}$)  parameterized by the population state $\bar{x}_n$. Let the $k^{th}$ element of the vector $H(G_n,x_n)$ be
	\begin{align}
		H_k(x_n,G_n) &= (1- x_n(k)) \frac{\nu k \theta_n^Z}{D} - x_n(k)\delta \quad \textit{where},\\
		\theta_n^Z &= \sum_{k}\bigg(\sum_{k'}P(k')P_{G_n}(k|k')\bigg) x_n(k).
	\end{align} Further, assume that $H(x,\mathcal{G}_i)$ is Lipschitz continuous in $x$ for all $\mathcal{G}_i\in \mathcal{G}$. Then, the sequence of the population state vectors $\{\bar{x}_n\}_{n\geq 0}$ of the of SIS diffusion over the reactive network converges weakly to the trajectory of the deterministic differential equation
	\begin{align}
		\label{eq:ODE_active_network}
		\frac{dx}{dt} &= \mathbb{E}_{G\sim \pi_G}\{H(x,G)\} \quad\text{with the algebraic constraint,}\\
		\label{eq:algebraic_constraint}
		P'_x\pi_x &= \pi_x.
	\end{align}
\end{theorem}

\begin{proof}
The following result from \cite{kushner2003} will be used to establish the weak convergence of the sequence of population states $\{\bar{x}_n\}_{n \geq 0}$ in Theorem \ref{th:active_network}.

Consider the stochastic approximation recursion,
\begin{equation}
\label{eq:stochastic_approximation_proof}
\bar{x}_{n+1} = \bar{x}_{n} + \epsilon \mathcal{H}(\bar{x}_n,G_n), \quad n = 0,1,\dots
\end{equation} where $\epsilon > 0$, $\{G_n\}$ is a $\mathcal{G}$ valued random process and, $\bar{x}_n \in \mathbb{R}^M$ is the output of recursion at time $n = 0,1,\dots$.
Further, let 
\begin{equation}
\label{eq:interpolated_trajectory}
\bar{x}^\epsilon(t) = \bar{x}_n \text{ for } t \in [n\epsilon,n\epsilon + \epsilon], \quad n = 0,1,\dots,
\end{equation}
which is a piecewise constant interpolation of $\{\bar{x}_n\}$.
In this setting, the following result holds. 

\begin{theorem}
	\label{th:SA}
	Consider the stochastic approximation algorithm (\ref{eq:stochastic_approximation_proof}). Assume
	\begin{compactenum}
		\item[{\bf SA1:}] $\mathcal{H}(x, G)$ is uniformly bounded for all $x \in \mathbb{R}^M$ and $G \in \mathcal{G}$.
		
		\item[{\bf SA2:}] For any $l \geq 0$, there exists $h(x)$ such that
		\begin{equation}
		\frac{1}{N} \sum_{n = l}^{N+l-1} \mathbb{E}_l\{\mathcal{H}(x, G_n)\} \rightarrow h(x) \text{ as } N \rightarrow \infty.
		\end{equation} where, $\mathbb{E}_l \{\cdot\}$ denotes expectation with respect to the sigma algebra generated by ${\{G_n: n < l\}}$.\label{th:SA_condition_2}
		
		\item[{\bf SA3:}] The ordinary differential equation (ODE)
		\begin{equation}
		\label{eq:ODE}
		\frac{dx(t)}{dt} = h(x(t)),\quad  x(0) = \bar{x}_0
		\end{equation} has a unique solution for every initial condition. \label{th:SA_condition_3}
	\end{compactenum}

	Then, the interpolated estimates  $\theta^\epsilon(t)$ defined in (\ref{eq:interpolated_trajectory}) satisfies
	\begin{equation}
	\lim_{\epsilon \to 0} \mathbb{P}\big( \sup_{0 \leq t \leq T}  \vert \bar{x}^\epsilon(t) - x(t)  \vert \geq \eta \big) = 0 \text{ for all } T>0, \eta > 0
	\end{equation} where, $(t)$ is the solution of the ODE (\ref{eq:ODE}).
\end{theorem}

Next, we will use Theorem \ref{th:SA} to show how the dynamics of the population state can be approximated by and ODE with an algebraic constraint in the case of a reactive network. 

By Part 2 of Theorem \ref{th:MFD}, the stochastic dynamics of the state $\bar{x}_n$ can be replaced by their mean-field dynamics $x_n$ as follows:
\begin{align}
\label{eq_proof:MFD_active_network}
x_{n+1} = x_n + H(x_n,G_n) 
\end{align}
where $H(x_n, G_n)$ is as defined in Theorem \ref{th:active_network}. Note that (\ref{eq_proof:MFD_active_network}) resembles (\ref{eq:stochastic_approximation_proof}).

\vspace{0.2cm}
\noindent
{\bf SA1 condition} - Each element $H_k(x,G)$ of $H(x,G)$ (for any $x, G$ in the domain) is a difference of two values(each in the interval $[0,1]$). Hence, SA1 condition holds. 

\vspace{0.2cm}
\noindent
{\bf SA2 condition} - As a result of the law of large numbers of the Markovian graph process $\{G_n\}$, SA 2 holds with
\begin{equation}
h(x) = \mathbb{E}_{G\sim \pi_x}\{H(x,G)\} 
\end{equation}
where,
$\pi_x$ is the unique stationary distribution satisfying $P'_x\pi_x = \pi_x$.

\vspace{0.2cm}
\noindent
{\bf SA3 condition} - Lipschitz continuity of $h(x)$ is a sufficient condition for the existence of a unique solution for a non-linear ODE. Hence, SA3 condition holds. 

Therefore, the result follows from Theorem \ref{th:SA}.
\end{proof}

\noindent
{\bf Importance of Theorem \ref{th:active_network} from a statistical modeling and signal processing perspective:} Theorem \ref{th:active_network} shows that the dynamics of the population state of the SIS diffusion over a reactive network can be approximated by an ODE with an algebraic constraint. From a statistical modeling perspective, Theorem \ref{th:active_network} provides a useful means of approximating the complex dynamics of two interdependent stochastic processes (diffusion process and the stochastic graph process) by an ODE (\ref{eq:ODE_active_network}) whose trajectory $x(t)$ at each time instant $t>0 $ is constrained by the algebraic condition (\ref{eq:algebraic_constraint}). Further, having an algebraic constraint restricts the number of possible sample paths of the population state vector $\{\bar{x}_n\}_{n\geq 0}$. Hence, from a statistical inference/filtering perspective, this makes estimation/prediction of the population state easier. For example, the algebraic condition can be used in Bayesian filtering algorithms (such as the one proposed in \cite{krishnamurthy2017tracking}) for the population state  to obtain more accurate results.

\section{Conclusion}
This paper explored the SIS contagion processes over social networks using a discrete-time model where, a randomly sampled node (at each time instant) faces the decision of changing its state (infected or susceptible). The mean-field dynamics was adopted to approximate the dynamics due to the exponentially large state space of the contagion process. It was shown that distributions with which the evolving node is chosen lead to different dynamics, but they induce the same critical threshold on model parameters that decides whether the contagion will spread or die out. Further, it was shown that monophilic adoption rules (taking a decision by observing friends of friends) make it easier for a contagion to spread instead dying out as a result of friendship paradox. The disassortativity of the network further amplifies the effect of the friendship paradox. Finally, the case where the underlying network is a reactive network that randomly evolves depending on the state of the contagion was studied. It was shown that the complex collective dynamics of the two (functionally) dependent stochastic processes (SIS contagion and the random graph process) can be approximated by a simple deterministic ODE whose trajectory satisfies an algebraic constraint. Our results shed light on how graph theoretic and sociological concepts such as friendship paradox and weak-ties affect dynamical processes (contagions) over social networks and, provide simple deterministic models to approximate the complex collective dynamics of contagions over stochastic graph processes.

% or


%\appendices
%\section{Proof of Theorem \ref{thm:MFD_samping_YZ}}
%\label{appn:MFD_samping_YZ}
%
%% you can choose not to have a title for an appendix
%% if you want by leaving the argument blank
%\section{Proof of Theorem \ref{thm:observe_XZ}}
%\label{appn:observe_XZ}
%\subsection*{Part 1: Non-Monophilic Adoption Rule}
%
%
%\section{Proof of Theorem \ref{th:active_network}}
%\label{appn:active_network}

\ifCLASSOPTIONcompsoc
%\else
%  % regular IEEE prefers the singular form
%\fi

\ifCLASSOPTIONcaptionsoff
  \newpage
\fi



\bibliographystyle{IEEEtran}
\bibliography{TNSE_2018_REF}
\end{document}